\documentclass[journal]{IEEEtran}

\usepackage[T1]{fontenc}
\usepackage[utf8]{inputenc}
\usepackage{amsmath,amssymb,amsfonts,amsthm,mathtools}
\usepackage{bm}
\usepackage{cite}
\usepackage{booktabs}
\usepackage{graphicx}
\usepackage{enumitem}
\usepackage{url}
\usepackage{hyperref}
\usepackage{placeins}
\usepackage[colorinlistoftodos]{todonotes}
\hypersetup{hidelinks}

\newtheorem{theorem}{Theorem}

\newtheorem{proposition}{Proposition}
\newtheorem{corollary}{Corollary}
\newtheorem{assumption}{Assumption}
\newtheorem{definition}{Definition}
\newtheorem{problem}{Problem}
\newtheorem{remark}{Remark}

\newcommand{\R}{\mathbb{R}}
\newcommand{\E}{\mathbb{E}}
\newcommand{\He}{\operatorname{He}}
\newcommand{\diag}{\operatorname{diag}}
\newcommand{\Tr}{\operatorname{tr}}
\newcommand{\sat}{\operatorname{sat}}

\newcommand{\pf}{\operatorname{pf}}
\newcommand{\argmin}{\operatorname*{arg\,min}}

\newcommand{\norm}[1]{\left\lVert #1\right\rVert}

\newcommand{\mcA}{\mathcal{A}}
\newcommand{\mcC}{\mathcal{C}}

\newcommand{\mcF}{\mathcal{F}}

\newcommand{\mcU}{\mathcal{U}}
\newcommand{\mcX}{\mathcal{X}}

\title
{Implicit Neural Networks as Static
Controllers: Certificates and Performance Separation}

\author{Giuseppe C. Calafiore and Laurent El Ghaoui}

\begin{document}
\maketitle


\begin{abstract}
Implicit neural controllers (INCs) are static feedback laws that are evaluated through an algebraic fixed point {equation}; they include as special cases neural network controllers. We propose a so-called implicit representation of neural networks as a key enabling device that exposes the controller as a trainable linear interconnection closed through a known static activation map, thereby making well-posedness and Lyapunov/IQC analysis mathematically easy to handle. For finite-dimensional LTI plants, we first develop a rigorous analysis theory for a given INC, including Perron--Frobenius and norm conditions for well posedness, LMI/IQC certificates for exponential stability, and LMIs for discounted infinite-horizon quadratic performance. 
We then formulate synthesis as a certification-compatible heuristic search: training is carried out under explicit well-posedness constraints, implicit-differentiation formulas provide gradients, and the resulting controller is accepted only after independent post-training LMIs or regional admissibility checks are feasible.  Finally, we establish constrained-control separation results: for a specific scalar unstable plant with hard actuator bounds,  an INC achieves a strictly smaller discounted infinite-horizon cost than any admissible finite-order dynamic linear controller. Additional results cover quadratic state-input costs, comparison with linear static output feedback, and computable upper/lower-bound certificates. Numerical examples illustrate the mechanism and the resulting certified performance.
\end{abstract}


\section{Introduction}
\IEEEPARstart{U}{sing} neural networks as controllers for dynamical systems has a long history. However, so far, providing explicit stability and performance guarantees for neural control remains a challenge. A feedback law obtained by simulation or automatic differentiation is not easy to analyze mathematically; this gap is the main obstacle addressed in this paper.

The difficulty is partly due to how neural networks are represented and parametrized. The standard layered notation for neural networks is ideal for implementation, but it does not easily lend itself to mathematical analysis. For feedback control synthesis, this is a serious limitation; therefore, a representation that simplifies the notation is the required mathematical entry point for certification-based neural controller synthesis.


A key message of the present paper is that implicit neural models, introduced in \cite{elghaoui2020implicit}, provide such a mathematically convenient representation of neural networks. A control law that uses such a model will be termed an implicit neural controller (INC). Rather than describing the neural map as a sequence of layers, an implicit model of a neural network defines a "hidden feature vector" as the solution of a fixed-point equation. The model maps an input $v$ to an output $z$ via the following equations:
\begin{subequations}
\label{eq:imodel_intro}
\begin{align}
        \eta &= \phi(A\eta+Bv), \\
        z &= C\eta+Dv,
\end{align}
\end{subequations}
where \(v\) is an external input, \(\eta\) is the "hidden feature vector," and \(\phi\) is a nonlinear "activation" vector-valued function such as a sigmoid, and matrices and vectors $A,B,C,D$ contain the controller parameters. Note that the "hidden" vector $\eta$ is only implicitly defined via a fixed-point equation. Perhaps surprisingly, the formalism contains standard feedforward deep networks as special cases, including convolutional, residual nets, LSTM, transformer and other architectures, see \cite{elghaoui2020implicit}. It also permits much more general network structures that have loops in the inference graph. The price of this additional modeling freedom is twofold. First, we need to guarantee that the solution to the fixed point equation exists and is unique, which will impose constraints on the network parameters. Second, the training problem is more complex, as gradient computation cannot be done via a simple backpropagation.

A control theorist will immediately observe that we recover the familiar state-space models of LTI systems, upon replacing the activation function $\phi$ by the integration operator in eq. \eqref{eq:imodel_intro}, see Figure~\ref{fig:inc_structure}. We also recover the classical models of uncertain systems used in robust control, when $\phi$ is replaced by an uncertainty block. 
This paper builds on this similarity to develop analysis and controller synthesis results for LTI systems with constraints. 


\begin{figure}[t]
\centering
\includegraphics[width=0.45\textwidth]
{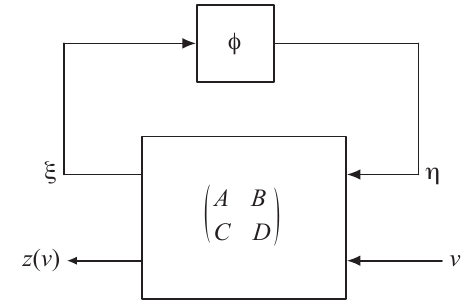}
\caption{State-space models for LTI systems and implicit models representations of neural networks use the same structural idea: a linear interconnection is closed through a  block $\phi$. In an LTI system the closing block $\phi$ is an integral or a unit delay operator, whereas in an implicit model $\phi$ is a static activation map. This similarity is what makes Lyapunov/IQC analysis for neural control possible.}
\label{fig:inc_structure}
\end{figure}


The paper is organized around three questions. First, given an implicit neural controller, can we certify that it is well posed, stabilizing, and performant for the given plant? Second, how can one design such a controller via a training algorithm? Third, can one prove that a simple neural static feedback law may provide a constrained-control advantage over linear feedback?

\subsection{Contributions}
\noindent
Our main message is that the implicit representation turns a static neural controller into a certificate-ready feedback interconnection, and that actuator constraints create settings in which neural feedback has a provable advantage over every admissible finite-order dynamic linear controller. The main contributions are as follows.

\begin{enumerate}[leftmargin=*]
\item \textbf{Implicit state-space representation and analysis.} Equipped with a static neural feedback formulated in the fixed-point form \eqref{eq:imodel_intro}, we derive Perron--Frobenius and norm-based well-posedness tests, Lipschitz bounds, and LMI/IQC certificates for exponential stability and discounted infinite-horizon quadratic performance.
 These certificates are first stated for the origin and zero biases. We then give an incremental version around an arbitrary equilibrium, thereby covering biased ReLU and saturation realizations. 
 
\item \textbf{Certification-compatible training algorithm.} We provide a nonconvex, heuristic training algorithm in which projected updates preserve well-posedness, implicit differentiation supplies gradients through a fixed point mechanism, and a train--certify--retrain gate accepts a controller only after a post-training LMI or independent regional check succeeds.

\item \textbf{Comparison with linear controllers.} We prove strict long-run constrained separation from all finite-order dynamic linear controllers for a scalar unstable plant, both for discounted state cost and for a range of quadratic state-input costs. We also prove separation from linear static output feedback, explain why no strict dominance statement can hold against arbitrary nonlinear static output feedback, and give computable upper/lower-bound certificates for general LTI comparisons.

\end{enumerate}
\subsection{Related Work}
Early neuro-control work studied neural identifiers, adaptive neural controllers, and neural approximations of nonlinear feedback laws for systems whose dynamics are hard to model exactly \cite{narendra1990identification,miller1995neural,lewis1999neural}. More recent surveys and tutorials emphasize that automatic differentiation and modern machine-learning software make both discrete-time and continuous-time neural controllers easier to optimize, including neural ODE controllers, actor-only model-based control, and comparisons with MPC \cite{boettcher2026control}. Deep reinforcement learning provides another broad route to neural feedback design, especially when accurate differentiable models are unavailable, but certification remains difficult for learned policies \cite{mnih2015human,degrave2022magnetic}.

A second line of work uses neural networks inside optimization and predictive-control architectures. Neural-network MPC dates back to early model-predictive schemes with neural plant models \cite{saintdonat1991neural,draeger1995mpc}, and recent differentiable-control layers include OptNet and differentiable MPC \cite{amos2017optnet,amos2018differentiablempc}. Neural ODEs and related differentiable dynamical models make it possible to learn or optimize through continuous-time dynamics \cite{chen2018neuralode}, while sparse identification and physics-informed or hybrid models have been used to support MPC in low-data regimes \cite{kaiser2018sindy}. These methods are primarily design tools; the present paper focuses on a complementary question, namely what can be certified once a neural feedback law has been fixed.

The implicit representation used here comes from implicit deep learning, where prediction rules are written as fixed-point equations and analyzed for well-posedness, robustness, sparsity, and trainability,  see \cite{elghaoui2020implicit}. It is also related to deep equilibrium models, differentiable optimization layers, and other implicit layers in which the forward pass is defined as the solution of an equation or an optimization problem \cite{bai2019deep,amos2017optnet}. Our contribution is to use this representation as a control-theoretic object: the neural controller is a linear interconnection closed through a known static nonlinearity.

The certification tools are rooted in absolute stability, Lyapunov theory, the S-procedure, linear matrix inequalities, and integral quadratic constraints \cite{boyd1994lmi,megretski1997iqc,khalil2002nonlinear,dullerud2000robust}. large modern literature applies convex relaxations, quadratic constraints, and semidefinite programming to verify neural networks and neural-network-controlled systems \cite{raghunathan2018sdp,fazlyab2022safety}. Closest in spirit are quadratic-constraint certificates for neural-network controllers \cite{yin2022qc} and dissipativity-constrained neural-controller synthesis \cite{junnarkar2024dissipativity}. Our approach follows the same Lyapunov/IQC philosophy, but it exploits the implicit controller form directly: the activation appears as an explicit sector-bounded block, while well-posedness of the hidden algebraic loop is certified separately.

Relative to this literature, the novelty is representation-centered. The controller is written in a realization form that exposes the hidden algebraic loop, gives a separate well-posedness test, permits hidden feedback during inference, and interfaces directly with plant-level Lyapunov certificates. The constrained separation theorem is a further analytical contribution: it identifies a hard-input mechanism by which saturated neural feedback can outperform globally admissible linear feedback.

Finally, the separation results are connected to constrained linear control,  invariant-set methods, and explicit MPC \cite{gilbert1991linear,tarbouriech2011stability,bemporad2002explicit,rawlings2017mpc}. Explicit MPC shows that constrained linear-quadratic problems naturally generate continuous piecewise-affine feedback laws. The present paper does not try to approximate MPC; instead, it isolates a simple analytical mechanism by which a saturated ReLU/INC law can be certified to outperform globally admissible linear feedback.

\subsection{Organization}
Section~\ref{sec:implicit_language} presents the implicit representation as a state-space language for neural controllers. Section~\ref{sec:problem} specializes it to the plant--INC interconnection. Section~\ref{sec:analysis} contains the first technical core of the paper, namely the analysis of a given INC: well-posedness, stability, and performance. Section~\ref{sec:training} describes a certification-compatible training framework. Section~\ref{sec:separation} contains the second technical core, namely separation and comparison results with linear controllers. Section~\ref{sec:numerics} gives numerical examples, and the final section concludes with limitations and open questions.

\subsection{Notation}
For a matrix \(M\), \(|M|\) denotes componentwise absolute value and \(\norm{M}_\infty=\max_i\sum_j |M_{ij}|\). If \(M\geq0\), \(\lambda_{\pf}(M)\) is its Perron--Frobenius eigenvalue. The notation \(P\succ0\) means that \(P\) is symmetric positive definite, and \(\He(Y)=Y+Y^\top\). The ReLU is \(\sigma(r)=\max\{0,r\}\), and the leaky ReLU with negative slope \(\ell\in[0,1]\) is \(\sigma_\ell(r)=\max\{r,\ell r\}\). The saturation is \(\sat(r)=\max\{-1,\min\{r,1\}\}\). For \(Q\succeq0\), \(\norm{x}_Q^2=x^\top Qx\).

\section{Implicit Neural Controllers}
\label{sec:implicit_language}
The implicit representation \eqref{eq:imodel_intro} has three ingredients: a linear interconnection parametrized by matrices; a fixed nonlinear map \(\phi\), and a hidden computational variable \(\eta\). This structure is the key feature that makes the model useful for control. The model's matrices are the trainable part; the activation is known; and the fixed-point equation exposes the exact algebraic loop that must be well posed before the map can be used as a feedback-control device.

Ordinary feedforward deep neural networks are special cases of our setup. For a layered network \(h_{\ell+1}=\phi_\ell(W_\ell h_\ell)\), one stacks the layer activations into \(\eta\). The resulting matrix \(W\) is strictly block triangular, so the fixed-point equation can be solved by backward substitution, which is just the usual forward pass. Hence feedforward architectures are implicit models with a nilpotent hidden interconnection. Allowing a more general matrix \(W\) removes the artificial requirement that hidden features be arranged in layers. The resulting model can contain loops in its inference graph, provided the fixed point remains unique.

This last proviso is essential:  in contrast with a feedforward network, an implicit neural model may have no fixed point or multiple fixed points for the same input. Well-posedness is therefore the neural analogue of requiring that an algebraic feedback interconnection define a unique input-output map. The Perron--Frobenius condition used below is a sufficient test that guarantees that the fixed point exists, is unique, and can be computed via fixed-point iteration.

The same representation also explains why sector and IQC tools enter naturally. For ReLU, leaky ReLU, saturated-linear, and related activations, the activation input-output pair \((z,\phi(z))\), equivalently \((z,\eta)\) on the graph \(\eta=\phi(z)\), satisfies pointwise quadratic inequalities. Once the controller is written as a linear interconnection plus \(\phi\), these inequalities become static IQCs for the closed-loop plant-controller system. The LMI conditions derived later are therefore the standard Lyapunov/IQC mechanism applied to an implicit neural controller.

In the remainder of the paper, \eqref{eq:imodel_intro}  is used with the plant measurement as input and the actuator signal as output. This specialization  gives a clean setting in which the representational issue, the well-posedness issue, and the constrained-performance advantage can all be stated and proved exactly.

\section{Closed-Loop Setting: Plant and INC Interconnection}
\label{sec:problem}
We now instantiate the representation of Section~\ref{sec:implicit_language} as a static feedback law for an LTI plant. 
%
Consider the discrete-time LTI plant
\begin{equation}
\label{eq:plant_dt}
        x_{k+1}=A_px_k+B_pu_k,\qquad y_k=C_px_k,
\end{equation}
where \(x_k\in\R^n\), \(u_k\in\R^m\), and \(y_k\in\R^p\). The full-state case is obtained by setting \(C_p=I\). We also consider the continuous-time counterpart
\begin{equation}
\label{eq:plant_ct}
        \dot x=A_px+B_pu,
        \qquad y=C_px.
\end{equation}
The static implicit controller is
\begin{subequations}
\label{eq:controller}
\begin{align}
        \eta_k &= \phi(W\eta_k+Uy_k+b), \label{eq:controller_fixed}\\
        u_k &= Ky_k+V\eta_k+d.          \label{eq:controller_output}
\end{align}
\end{subequations}
The hidden vector \(\eta_k\in\R^q\) is recomputed from \(y_k\) at each time. In Sections~\ref{sec:stability}--\ref{sec:performance}, we take \(b=d=0\) and assume \(\phi(0)=0\), so that the origin is a closed-loop equilibrium. Biases and nonzero equilibria are handled by the incremental certificate in Section~\ref{sec:incremental}; the scalar saturation realization in Appendix~\ref{app:saturation_relu} uses biases explicitly.

Figure~\ref{fig:plant_inc_interconnection} shows the resulting closed-loop interconnection. It also indicates the two roles played by the equations in this section. For analysis, the plant and the INC matrices are fixed and one asks whether the interconnection is well posed, stable, and performant. During training, the plant is fixed and the INC matrices are varied numerically; after training, however, the resulting controller is again treated as fixed and must pass the independent certificates.

\begin{figure}[t]
\centering
\includegraphics[width=\linewidth]{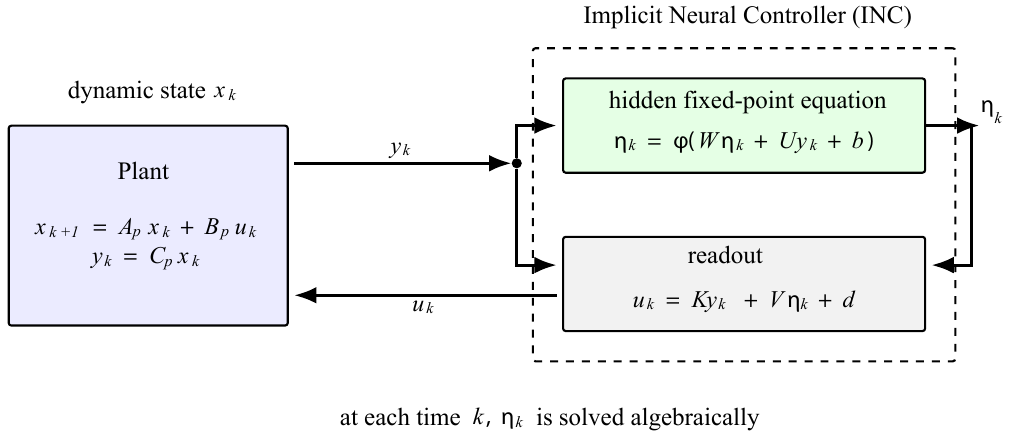}
\caption{Feedback interconnection of the LTI plant \eqref{eq:plant_dt} and the INC \eqref{eq:controller}. Observe that the hidden variable \(\eta_k\) is {\em not} a dynamic controller state; it is the solution of a static fixed-point equation at time \(k\).}
\label{fig:plant_inc_interconnection}
\end{figure}

\begin{assumption}[Componentwise nonexpansive activation]
\label{ass:cone}
The activation \(\phi:\R^q\to\R^q\) acts componentwise and satisfies
\begin{equation}
\label{eq:cone}
        |\phi(r)-\phi(s)|\leq |r-s|,
        \qquad \forall r,s\in\R^q,
\end{equation}
with componentwise absolute values and inequalities.
\end{assumption}

ReLU, leaky ReLU with slope in \([0,1]\), tanh, and saturated-linear activations satisfy Assumption~\ref{ass:cone}. The logistic sigmoid is nonexpansive after scaling but is not origin preserving unless it is centered; such biased or nonzero-equilibrium uses are covered by the incremental condition below.

\begin{assumption}[Static sector condition]
\label{ass:sector}
The activation satisfies \(\phi(0)=0\). There exist diagonal matrices
\begin{equation}
\begin{aligned}
        D_0&=\diag(\underline m_i),\qquad
        D_1=\diag(\overline m_i),\\
        0&\leq \underline m_i\leq \overline m_i\leq1.
\end{aligned}
\end{equation}
such that, for every component,
\begin{equation}
\label{eq:sector_general}
        \big(\phi_i(z_i)-\underline m_i z_i\big)
        \big(\overline m_i z_i-\phi_i(z_i)\big)\geq0,
        \qquad i=1,\ldots,q.
\end{equation}
\end{assumption}

Assumption~\ref{ass:sector} is the componentwise sector \([D_0,D_1]\). The special case \(D_0=0\), \(D_1=I\) is the usual sector \([0,1]\) used for ReLU and saturation. A leaky ReLU with negative slope \(\ell\) satisfies it with \(D_0=\ell I\) and \(D_1=I\).

\begin{assumption}[Incremental sector condition]
\label{ass:incremental_sector}
There exist the same diagonal matrices \(D_0,D_1\) as in Assumption~\ref{ass:sector} such that, for all \(r,s\in\R^q\),
\begin{equation}
\label{eq:incremental_sector}
\begin{aligned}
&\big(\phi_i(r_i)-\phi_i(s_i)-\underline m_i(r_i-s_i)\big)\\
&\quad\times
\big(\overline m_i(r_i-s_i)-\phi_i(r_i)+\phi_i(s_i)\big)\geq0,
\end{aligned}
\end{equation}
for \(i=1,\ldots,q\).
\end{assumption}

The incremental condition holds for componentwise slope-restricted nonlinearities. It is the version needed when biases are present or when the certified equilibrium is not the origin.
\section{Analysis of a Given INC}
\label{sec:analysis}
The first key issue is analysis. Once the plant and controller matrices are fixed, the controller must define a unique static map, and the closed loop must admit verifiable stability and performance guarantees. This section treats these three questions in order.

\subsection{Well-Posedness of the Controller Map}
\label{sec:wellposed}
The equation \(\eta=\phi(W\eta+Uy+b)\) must define a unique \(\eta\) for each measurement. The next theorem is the controller version of the Perron--Frobenius fixed-point test for implicit models.

\begin{definition}[Well-posed implicit controller]
The matrix \(W\) is well posed for \(\phi\) if, for every \(c\in\R^q\), the equation
\begin{equation}
        \eta=\phi(W\eta+c)
\end{equation}
has a unique solution.
\end{definition}

\begin{theorem}[Perron--Frobenius well-posedness]
\label{thm:wellposed}
Suppose Assumption~\ref{ass:cone} holds. If
\begin{equation}
\label{eq:pf_condition}
        \lambda_{\pf}(|W|)<1,
\end{equation}
then the fixed-point equation in \eqref{eq:controller_fixed} has a unique solution for every \(y\) and every bias \(b\). Moreover, for the corresponding hidden map \(\eta(y)\),
\begin{equation}
\label{eq:hidden_lipschitz}
        |\eta(y_1)-\eta(y_2)|
        \leq (I-|W|)^{-1}|U|\,|y_1-y_2|.
\end{equation}
Consequently, the controller map \(\pi(y)=Ky+V\eta(y)+d\) is globally Lipschitz and
\begin{equation}
\label{eq:controller_lipschitz}
        |\pi(y_1)-\pi(y_2)|
        \leq \left(|K|+|V|(I-|W|)^{-1}|U|\right)|y_1-y_2|.
\end{equation}
\end{theorem}

\begin{proof}
For fixed \(y\), set \(c=Uy+b\). By Assumption~\ref{ass:cone},
\[
 |\phi(W\eta_1+c)-\phi(W\eta_2+c)|\leq |W|\,|\eta_1-\eta_2|.
\]
Since \(\lambda_{\pf}(|W|)<1\), the map is a contraction in a suitable weighted infinity norm. Hence the fixed point exists and is unique. For two measurements, subtracting the two fixed-point equations gives
\[
 |\eta(y_1)-\eta(y_2)|\leq |W|\,|\eta(y_1)-\eta(y_2)|+|U|\,|y_1-y_2|.
\]
The matrix \((I-|W|)^{-1}=\sum_{j=0}^{\infty}|W|^j\) is componentwise nonnegative, which proves \eqref{eq:hidden_lipschitz}. The output bound follows immediately from \eqref{eq:controller_output}.
\end{proof}

\begin{remark}[Convex sufficient condition]
Since \(\lambda_{\pf}(|W|)\leq\norm{W}_\infty\), the row-sum constraint
\begin{equation}
        \norm{W}_\infty\leq\kappa<1
\end{equation}
is a convex sufficient condition for well-posedness. This condition is particularly convenient in projected-gradient training.
\end{remark}

\begin{remark}[Feedforward networks as implicit controllers]
Any feedforward ReLU controller can be written as \eqref{eq:controller}: stack the layer activations into \(\eta\), and let \(W\) be strictly block triangular. Then \(\lambda_{\pf}(|W|)=0\), so the controller is well posed. Thus the implicit class contains standard static ReLU controllers and extends them by allowing well-posed hidden feedback.
\end{remark}

\subsection{Stability Certificates}
\label{sec:stability}
We now set \(b=d=0\). Define the stacked variable
\begin{equation}
\label{eq:s_stack}
        s=\begin{bmatrix}x\\ \eta\end{bmatrix},\qquad
        X=\begin{bmatrix}I_n&0\end{bmatrix},\qquad
        H=\begin{bmatrix}0&I_q\end{bmatrix}.
\end{equation}
Then \(x=Xs\), \(\eta=Hs\), and
\begin{equation}
\label{eq:Z_def}
        z=W\eta+UC_px=Zs,
        \qquad Z=\begin{bmatrix}UC_p&W\end{bmatrix}.
\end{equation}
The discrete-time closed-loop state update is
\begin{equation}
\label{eq:F_def}
        x^+=Fs,
        \qquad
        F=\begin{bmatrix}A_p+B_pKC_p&B_pV\end{bmatrix},
\end{equation}
and the control is
\begin{equation}
\label{eq:G_def}
        u=Gs,
        \qquad
        G=\begin{bmatrix}KC_p&V\end{bmatrix}.
\end{equation}

For a diagonal multiplier \(\Lambda\succeq0\), define
\begin{equation}
\label{eq:M_lambda}
\begin{aligned}
        M_\Lambda={}&-2Z^\top D_0\Lambda D_1Z
        +\He\!\left(H^\top\Lambda(D_0+D_1)Z\right)\\
        &\quad -2H^\top\Lambda H .
\end{aligned}
\end{equation}
For every \(s\) satisfying \(\eta=\phi(z)\), Assumption~\ref{ass:sector} gives
\begin{equation}
\label{eq:iqc_nonnegative}
        s^\top M_\Lambda s
        =2(\eta-D_0z)^\top\Lambda(D_1z-\eta)\geq0.
\end{equation}
This is a static IQC for the graph of the activation. Equation~\eqref{eq:M_lambda} reduces to the usual \([0,1]\)-sector multiplier when \(D_0=0\) and \(D_1=I\).

\begin{theorem}[Discrete-time exponential stability]
\label{thm:dt_stability}
Suppose Assumptions~\ref{ass:cone} and \ref{ass:sector} hold and \(\lambda_{\pf}(|W|)<1\). If there exist \(P\succ0\), a diagonal \(\Lambda\succeq0\), and \(\rho\in[0,1)\) such that
\begin{equation}
\label{eq:dt_stability_lmi}
        F^\top P F-\rho^2X^\top PX+M_\Lambda\preceq0,
\end{equation}
then the origin of \eqref{eq:plant_dt} in feedback with the bias-free controller is globally exponentially stable. More precisely,
\begin{equation}
        x_k^\top Px_k\leq \rho^{2k}x_0^\top Px_0,
        \qquad k=0,1,2,\ldots .
\end{equation}
\end{theorem}

\begin{proof}
Well-posedness follows from Theorem~\ref{thm:wellposed}. Let \(V_P(x)=x^\top Px\). Along the closed-loop graph,
\[
 V_P(x^+)-\rho^2V_P(x)=s^\top(F^\top PF-\rho^2X^\top PX)s.
\]
By \eqref{eq:dt_stability_lmi} and \eqref{eq:iqc_nonnegative},
\[
        V_P(x^+)-\rho^2V_P(x)\leq -s^\top M_\Lambda s\leq0.
\]
Iteration proves the claim.
\end{proof}

For continuous time, define
\begin{equation}
\label{eq:Fc_def}
        F_c=\begin{bmatrix}A_p+B_pKC_p&B_pV\end{bmatrix},
        \qquad \dot x=F_cs.
\end{equation}

\begin{theorem}[Continuous-time exponential stability]
\label{thm:ct_stability}
Suppose Assumptions~\ref{ass:cone} and \ref{ass:sector} hold and \(\lambda_{\pf}(|W|)<1\). If there exist \(P\succ0\), diagonal \(\Lambda\succeq0\), and \(\alpha>0\) such that
\begin{equation}
\label{eq:ct_stability_lmi}
        \He(X^\top PF_c)+\alpha X^\top PX+M_\Lambda\preceq0,
\end{equation}
then the origin of \eqref{eq:plant_ct} in feedback with the bias-free controller is globally exponentially stable and
\begin{equation}
        x(t)^\top Px(t)\leq e^{-\alpha t}x(0)^\top Px(0).
\end{equation}
\end{theorem}

\begin{proof}
The derivative of \(V_P\) satisfies \(\dot V_P=s^\top\He(X^\top PF_c)s\). The result follows from \eqref{eq:ct_stability_lmi}, \eqref{eq:iqc_nonnegative}, and Gronwall's inequality.
\end{proof}

\begin{remark}[Certification versus design]
For fixed controller matrices and fixed decay rate, \eqref{eq:dt_stability_lmi} and \eqref{eq:ct_stability_lmi} are LMIs in \((P,\Lambda)\). Joint optimization over controller parameters and Lyapunov variables is nonconvex, as expected for nonlinear controller design.
\end{remark}

\subsection{Discounted Quadratic Performance}
\label{sec:performance}
Let \(Q\succeq0\), \(R\succeq0\), and \(\beta\in(0,1]\). Consider the discounted cost
\begin{equation}
\label{eq:disc_cost}
        J_\beta(x_0)=\sum_{k=0}^{\infty}\beta^k
        \left(x_k^\top Qx_k+u_k^\top Ru_k\right).
\end{equation}

\begin{theorem}[Discounted performance certificate]
\label{thm:performance}
Suppose Assumptions~\ref{ass:cone} and \ref{ass:sector} hold and \(\lambda_{\pf}(|W|)<1\). If there exist \(P\succ0\) and diagonal \(\Lambda\succeq0\) such that
\begin{equation}
\label{eq:performance_lmi}
        \beta F^\top P F-X^\top P X+X^\top QX+G^\top RG+M_\Lambda\preceq0,
\end{equation}
then every closed-loop trajectory satisfies
\begin{equation}
\label{eq:cost_bound}
        J_\beta(x_0)\leq x_0^\top Px_0.
\end{equation}
Consequently, if \(S_0=\E[x_0x_0^\top]\) is the second-moment matrix of the initial condition, then
\begin{equation}
        \E[J_\beta(x_0)]\leq \Tr(PS_0).
\end{equation}
Equivalently, if \(x_0\) has mean \(\mu_0\) and centered covariance \(\Sigma_0\), the bound is \(\Tr(P\Sigma_0)+\mu_0^\top P\mu_0\).
\end{theorem}

\begin{proof}
For every closed-loop pair \((x,\eta)\), \eqref{eq:performance_lmi} and \eqref{eq:iqc_nonnegative} imply
\begin{equation}
        \beta V_P(x^+)-V_P(x)+x^\top Qx+u^\top Ru\leq0.
\end{equation}
Multiplying the inequality at time \(k\) by \(\beta^k\) and summing from \(0\) to \(N\) gives
\begin{align}
&\sum_{k=0}^{N}\beta^k(x_k^\top Qx_k+u_k^\top Ru_k)\notag\\
&\qquad\leq V_P(x_0)-\beta^{N+1}V_P(x_{N+1})\leq V_P(x_0).
\end{align}
Letting \(N\to\infty\) proves \eqref{eq:cost_bound}. Taking expectations gives the final statement.
\end{proof}

\begin{remark}[Use in controller comparison]
Theorem~\ref{thm:performance} gives an upper bound on the neural controller cost. To prove superiority over a benchmark class, one needs a lower bound on the benchmark cost. Sections~\ref{sec:scalar}--\ref{sec:general_certificate} provide such lower bounds for dynamic linear and linear static output-feedback controllers under hard input constraints.
\end{remark}

\subsection{Biases and Incremental Certificates}
\label{sec:incremental}
Theorems~\ref{thm:dt_stability}--\ref{thm:performance} were stated for a bias-free origin equilibrium. Practical neural controllers often use biases, and the saturated ReLU realization in Appendix~\ref{app:saturation_relu} is one such example. The following incremental statement is the corresponding certificate about an arbitrary equilibrium.

Let \((x_\star,\eta_\star,z_\star,u_\star)\) satisfy
\begin{subequations}
\label{eq:equilibrium_incremental}
\begin{align}
        x_\star&=A_px_\star+B_pu_\star,\qquad y_\star=C_px_\star,\\
        z_\star&=W\eta_\star+Uy_\star+b,\qquad
        \eta_\star=\phi(z_\star),\\
        u_\star&=Ky_\star+V\eta_\star+d .
\end{align}
\end{subequations}
For any trajectory define deviations \(\tilde x=x-x_\star\), \(\tilde\eta=\eta-\eta_\star\), \(\tilde z=z-z_\star\), \(\tilde u=u-u_\star\), and \(\tilde s=[\tilde x^\top\;\tilde\eta^\top]^\top\). The same matrices \(F,G,Z,X,H\) map deviations: \(\tilde x^+=F\tilde s\), \(\tilde u=G\tilde s\), and \(\tilde z=Z\tilde s\).

\begin{theorem}[Incremental stability and performance]
\label{thm:incremental_certificate}
Suppose Assumptions~\ref{ass:cone} and \ref{ass:incremental_sector} hold and \(\lambda_{\pf}(|W|)<1\). If \eqref{eq:dt_stability_lmi} is feasible, then \(x_\star\) is globally exponentially stable for the biased controller, with \(x-x_\star\) in place of \(x\). If \eqref{eq:performance_lmi} is feasible, then
\begin{equation}
\label{eq:incremental_cost_bound}
        \sum_{k=0}^{\infty}\beta^k
        \left(\tilde x_k^\top Q\tilde x_k+\tilde u_k^\top R\tilde u_k\right)
        \leq \tilde x_0^\top P\tilde x_0 .
\end{equation}
\end{theorem}

\begin{proof}
The incremental sector condition gives
\[
        2(\tilde\eta-D_0\tilde z)^\top\Lambda(D_1\tilde z-\tilde\eta)\geq0,
\]
which is exactly \(\tilde s^\top M_\Lambda\tilde s\geq0\). The proofs of Theorems~\ref{thm:dt_stability} and \ref{thm:performance} then apply verbatim to the deviation dynamics.
\end{proof}

\begin{remark}[Scope of the global certificates]
The LMIs above are global because the sector bounds are global. This is appropriate for slope-restricted activations and unconstrained inputs. Hard actuator constraints for unstable plants are inherently regional; in that setting the admissible initial set and the input constraint must be verified separately, as in the scalar theorem of Section~\ref{sec:scalar} or by a regional certificate.
\end{remark}

\section{Certification-Compatible Training Framework}
\label{sec:training}
The analysis results of Section~\ref{sec:analysis} are exact acceptance tests for a fixed controller, but they do not by themselves give a convex design problem. Indeed, the controller matrices enter the closed-loop matrices \(F,G,Z\), while the certificate variables \(P\) and \(\Lambda\) enter the LMIs; optimizing all of them jointly is a nonconvex problem. 
This section therefore describes a heuristic candidate-generation framework: training preserves well-posedness, produces controllers for subsequent certification, and leaves the final guarantee to a post-training test that is independent of the optimizer used to obtain the candidate.

\begin{problem}[Certification-compatible INC training]
\label{prob:certified_synthesis}
Given the plant \((A_p,B_p,C_p)\), an admissible initial set or distribution, and an actuator set \(\mcU\), find INC parameters \(\theta=(W,U,V,K,b,d)\) such that:
\begin{enumerate}[leftmargin=*]
\item the controller fixed point is well posed for every measurement in the design domain;
\item the controller respects the input constraints, either by construction or by an independent verification over the prescribed operating region;
\item the closed loop admits a stability certificate \eqref{eq:dt_stability_lmi}, an incremental certificate from Theorem~\ref{thm:incremental_certificate}, or a regional certificate;
\item the certified bound, or a clearly labeled simulation estimate when no certificate is claimed, is small.
\end{enumerate}
\end{problem}


\subsection{Simulation objective and well-posedness constraints}
Let \(\mu\) be a probability distribution on admissible initial conditions. A finite-horizon surrogate for the design objective is
\begin{equation}
\label{eq:training_problem}
\begin{aligned}
        \min_\theta\;&J_T(\theta):=\E_{x_0\sim\mu}\left[\sum_{k=0}^{T-1}\ell(x_k,\,\pi_\theta(C_px_k))+V_f(x_T)\right]\\
        \text{s.t. }&\eta_k=\phi(W\eta_k+Uy_k+b),\quad u_k=Ky_k+V\eta_k+d,\\
        &\norm{W}_\infty\leq\kappa<1.
\end{aligned}
\end{equation}
The same template can be used with a discounted loss when the target certificate is Theorem~\ref{thm:performance}. The constraint \(\norm{W}_\infty\leq\kappa<1\) is not the only possible well-posedness sufficient condition, but it is convex, simple to impose, and directly connected to Theorem~\ref{thm:wellposed}.

\begin{proposition}[Projected training preserves well-posedness]
\label{prop:projected_training_wp}
Assume the activation satisfies Assumption~\ref{ass:cone}. If every training update is followed by the row-wise projection
\begin{subequations}
\begin{align}
        W&\leftarrow \Pi_\kappa(W),\\
        \Pi_\kappa(W)&:=\argmin_{\norm{\bar W}_\infty\leq\kappa}
        \norm{\bar W-W}_F,
        \qquad \kappa<1,
\end{align}
\end{subequations}
then every controller produced during training is well posed. In particular, for every measurement \(y\), the fixed point \(\eta=\phi(W\eta+Uy+b)\) is unique.
\end{proposition}
\begin{proof}
The projection enforces \(\norm{W}_\infty\leq\kappa<1\). Since \(\lambda_{\pf}(|W|)\leq\norm{W}_\infty\), Theorem~\ref{thm:wellposed} applies at every iterate. The projection is row separable because \(\norm{W}_\infty\) is the maximum row \(\ell_1\)-norm.
\end{proof}

\subsection{Differentiating through the implicit controller}
Training \eqref{eq:training_problem} requires gradients through the algebraic inference map. The next proposition gives the formulas used in the candidate-generation loop.

\begin{proposition}[Implicit sensitivities and adjoints]
\label{prop:implicit_gradients}
Fix a measurement \(y\) and let \(\eta\) solve \(\eta=\phi(W\eta+Uy+b)\). Suppose \(\lambda_{\pf}(|W|)<1\) and let \(D_\phi=\diag(\delta_i)\), where \(\delta_i=\phi_i'(W\eta+Uy+b)_i\) at differentiability points; at ReLU kink points any \(\delta_i\in[0,1]\) gives a valid generalized derivative. Then \(I-D_\phi W\) is nonsingular and
\begin{equation}
\label{eq:inc_jacobian}
        \frac{\partial \eta}{\partial y}=(I-D_\phi W)^{-1}D_\phi U,
\end{equation}
so the local controller Jacobian is
\begin{equation}
\label{eq:inc_output_jacobian}
        \frac{\partial u}{\partial y}=K+V(I-D_\phi W)^{-1}D_\phi U.
\end{equation}
Moreover, for a scalar loss \(\ell(u)\), let \(q=\nabla_u\ell\). Define \(v\) as the unique solution of the adjoint fixed-point equation
\begin{equation}
\label{eq:adjoint_fixed_point}
        v=D_\phi(W^\top v+V^\top q).
\end{equation}
Then
\begin{subequations}
\label{eq:parameter_gradients}
\begin{align}
        \nabla_y\ell &= K^\top q+U^\top v,\qquad
        \nabla_W\ell = v\eta^\top,
        \qquad \nabla_U\ell = vy^\top,\\
        \nabla_b\ell &= v,
        \qquad \nabla_V\ell = q\eta^\top,
        \qquad \nabla_K\ell = qy^\top,
        \qquad \nabla_d\ell = q.
\end{align}
\end{subequations}
\end{proposition}
\begin{proof}
Differentiating \(\eta=\phi(W\eta+Uy+b)\) gives
\((I-D_\phi W)d\eta=D_\phi Udy\), which proves \eqref{eq:inc_jacobian}. Since \(0\preceq D_\phi\preceq I\) componentwise, the nonnegative-matrix monotonicity of the Perron--Frobenius eigenvalue gives \(\lambda_{\pf}(|D_\phi W|)\leq\lambda_{\pf}(|W|)<1\); hence \(I-D_\phi W\) is nonsingular. Transposing the sensitivity relation yields the reverse-mode equation \eqref{eq:adjoint_fixed_point}; the adjoint fixed point is unique by the same argument applied to \(|D_\phi W^\top|\). The parameter-gradient expressions then follow by writing the differential of \(u=Ky+V\eta+d\) and of the affine argument \(W\eta+Uy+b\).
\end{proof}

\subsection{Certificate-aware retraining and the acceptance gate}
After training, the candidate controller is certified by solving the relevant LMI with the controller matrices fixed. If the stability certificate is the target, define
\begin{equation}
\label{eq:stab_violation}
        \nu_{\rm stab}(\theta;P,\Lambda)=\lambda_{\max}\!\big(F^\top P F-\rho^2X^\top PX+M_\Lambda\big).
\end{equation}
If the performance certificate is the target, define
\begin{multline}
\label{eq:perf_violation}
        \nu_{\rm perf}(\theta;P,\Lambda)=\lambda_{\max}\!\big(\beta F^\top P F-X^\top PX+X^\top QX\\
        +G^\top RG+M_\Lambda\big).
\end{multline}
Feasibility corresponds to making the selected violation nonpositive for some \(P\succ0\) and diagonal \(\Lambda\succeq0\). If certification fails, one can return to training with a certificate-aware penalty such as
\begin{equation}
\label{eq:cert_aware_loss}
        J_T(\theta)+\gamma_{\rm cert}\max\{0,\nu_{\rm perf}(\theta;P,\Lambda)\},
\end{equation}
where \((P,\Lambda)\) may be periodically recomputed by solving the SDP for the current controller. 
This remains an heuristic, but it makes the optimization aware of the certificate that will ultimately be required.

\begin{proposition}[Certification gate]
\label{prop:certification_gate}
Suppose the training loop returns parameters \(\theta\) satisfying \(\lambda_{\pf}(|W|)<1\). If \eqref{eq:dt_stability_lmi} is feasible for the resulting closed loop, then the returned INC is well posed and exponentially stabilizing. If \eqref{eq:performance_lmi} is feasible, then it is well posed and satisfies the discounted-performance bound \eqref{eq:cost_bound}. These guarantees are independent of the optimizer used to obtain \(\theta\).
\end{proposition}
\begin{proof}
Well-posedness follows from Theorem~\ref{thm:wellposed}. The stability and performance conclusions are exactly Theorems~\ref{thm:dt_stability} and \ref{thm:performance}, applied after training with the matrices fixed.
\end{proof}

\noindent\textbf{Algorithm 1 (certification-compatible train--certify--retrain).}
\begin{enumerate}[leftmargin=*]
\item Choose an INC architecture, a well-posedness budget \(\kappa<1\), and, if hard input limits must hold globally, a saturated readout or another constraint-preserving parameterization.
\item Train by simulation using \eqref{eq:training_problem}, differentiating through the fixed point using Proposition~\ref{prop:implicit_gradients}, and projecting \(W\) as in Proposition~\ref{prop:projected_training_wp}.
\item With the trained controller fixed, solve the stability, incremental, or performance LMI.
\item If certification fails, retrain with a penalty such as \eqref{eq:cert_aware_loss}, tighten \(\kappa\), change the architecture, or initialize from a certified linear controller.
\item Accept the controller only when Proposition~\ref{prop:certification_gate} or an independent regional certificate applies.
\end{enumerate}

\subsection{Hard input constraints and regional checks}
Sample-wise clipping or projection of controls during training is useful as a heuristic, but it is not by itself a proof that \(u(x)\in\mcU\) for all states in a continuum. In this paper, hard input constraints are certified in one of two ways. First, the readout can impose the constraint by construction, for example through a componentwise saturation. Second, on a compact operating set \(\mcX\), one may independently verify
\begin{equation}
\label{eq:regional_input_check}
        \pi_\theta(C_px)\in\mcU,
        \qquad A_px+B_p\pi_\theta(C_px)\in\mcX,
        \qquad \forall x\in\mcX .
\end{equation}
For feedforward ReLU or saturated piecewise-affine controllers, \eqref{eq:regional_input_check} can be checked exactly by activation-region enumeration when the dimension is modest, or conservatively by mixed-integer, interval, or polyhedral bounds. For  implicit controllers with \(W\neq0\), the same verification must include the fixed-point equation; Theorem~\ref{thm:wellposed} supplies uniqueness, but not by itself input admissibility.

The training framework is certificate-compatible in the sense that learning supplies candidate controllers, while the implicit representation supplies well-posedness, differentiation, and post-training certificate machinery. When no LMI or regional check is supplied, the result should be described as simulation evidence rather than as a certificate.

\section{Separation and Comparison Results}
\label{sec:separation}
We now turn from analysis and training to comparison. The results in this section show that the INC class, under actuator constraints,  can achieve a provable performance separation from linear controllers. 

\subsection{Long-Run Separation from Dynamic Linear Control}
\label{sec:scalar}
We now prove that a static ReLU controller can outperform every finite-order dynamic linear controller in an infinite-horizon performance index. The proof is by construction: we exhibit a scalar constrained plant, derive a first-move lower bound for all admissible dynamic linear controllers, and compute the exact cost of a saturated two-ReLU law.

Consider the system
\begin{equation}
\label{eq:scalar_plant}
        x_{k+1}=\frac32 x_k+u_k,
        \qquad |u_k|\leq1,
        \qquad x_0\sim {\rm Unif}[-1,1].
\end{equation}
Let \(\mcC_{\rm dyn}\) be the class of finite-dimensional LTI dynamic controllers
\begin{subequations}
\label{eq:dyn_linear}
\begin{align}
        \xi_{k+1}&=A_c\xi_k+B_cx_k,\\
        u_k&=C_c\xi_k+D_cx_k,
\end{align}
\end{subequations}
with fixed initial controller state \(\xi_0\) independent of \(x_0\), satisfying \(|u_k|\leq1\) for all \(k\geq0\) and all \(x_0\in[-1,1]\). At \(k=0\), every such controller has
\begin{equation}
        u_0=d+\kappa x_0,
\end{equation}
with \(|d|+|\kappa|\leq1\).
Define the discounted state cost
\begin{equation}
\label{eq:J_beta_state}
        J_\beta(\pi)=\E\sum_{k=0}^{\infty}\beta^k x_k^2,
        \qquad \beta\in(0,1].
\end{equation}

\begin{theorem}[Infinite-horizon separation from dynamic linear control]
\label{thm:longrun_dyn}
For the system \eqref{eq:scalar_plant}, every controller in \(\mcC_{\rm dyn}\) satisfies
\begin{equation}
\label{eq:dyn_lower_state}
        J_\beta\geq \frac13+\frac{\beta}{12}.
\end{equation}
The static two-ReLU controller
\begin{equation}
\label{eq:relu_sat_controller}
        u=1-\sigma\!\left(\frac32x+1\right)+\sigma\!\left(\frac32x-1\right)
        =-\sat\!\left(\frac32x\right)
\end{equation}
satisfies \(|u|\leq1\) and has
\begin{equation}
\label{eq:nn_cost_state}
        J_\beta^{\rm NN}=\frac13+\frac{\beta}{36}.
\end{equation}
Therefore
\begin{equation}
        J_\beta^{\rm NN}<\inf_{\pi\in\mcC_{\rm dyn}}J_\beta(\pi)
        \qquad \forall\beta\in(0,1].
\end{equation}
The certified gap is at least \(\beta/18\).
\end{theorem}

\begin{proof}
For any dynamic linear controller, the first input satisfies \(|d|+|\kappa|\leq1\). Since \(x_0\) is symmetric with variance \(1/3\),
\begin{equation}
        \E[x_1^2]
        =\E\left[\left(\left(\frac32+\kappa\right)x_0+d\right)^2\right]
        =\frac{(\frac32+\kappa)^2}{3}+d^2.
\end{equation}
The minimum subject to \(|d|+|\kappa|\leq1\) is attained at \(d=0\), \(\kappa=-1\), and equals \(1/12\). Since all future terms in \eqref{eq:J_beta_state} are nonnegative,
\[
        J_\beta\geq \E[x_0^2]+\beta\E[x_1^2]
        \geq \frac13+\frac{\beta}{12}.
\]

For \eqref{eq:relu_sat_controller}, \(x_1=0\) when \(|x_0|\leq2/3\), and \(|x_1|=1.5|x_0|-1\) when \(2/3<|x_0|\leq1\). Hence
\begin{equation}
        \E[x_1^2]=\int_{2/3}^{1}(1.5x-1)^2\,dx=\frac{1}{36}.
\end{equation}
Moreover, \(|x_1|\leq1/2<2/3\), so the next control action is \(u_1=-1.5x_1\) and \(x_2=0\). Thus all state costs after \(k=1\) vanish, proving \eqref{eq:nn_cost_state}. The strict gap is \((1/3+\beta/12)-(1/3+\beta/36)=\beta/18\).
\end{proof}

\begin{remark}[Why controller memory does not help]
The memory \(\xi_k\) of a dynamic linear controller can affect inputs for \(k\geq1\), but its first move is affine in \(x_0\). The hard input constraint forces this affine map to have slope at most one in magnitude over \([-1,1]\). The ReLU controller uses slope \(-1.5\) near the origin and saturates near the boundary. The unavoidable first-step cost lower bound for dynamic linear control cannot be recovered later because the infinite-horizon cost is nonnegative term by term.
\end{remark}

\subsection{Quadratic input penalty}
The same idea yields a conventional state-input index inequality. Define
\begin{equation}
\label{eq:J_beta_r}
        J_{\beta,r}(\pi)=\E\sum_{k=0}^{\infty}\beta^k(x_k^2+r u_k^2),
        \qquad r\geq0.
\end{equation}

\begin{theorem}[State-input performance separation]
\label{thm:input_penalty_sep}
For the plant \eqref{eq:scalar_plant} and \(\beta\in(0,1]\), the controller \eqref{eq:relu_sat_controller} satisfies
\begin{equation}
\label{eq:nn_input_penalty_cost}
        J_{\beta,r}^{\rm NN}
        =\frac13+\frac{\beta}{36}
        +r\left(\frac59+\frac{\beta}{16}\right).
\end{equation}
Every admissible finite-dimensional dynamic linear controller satisfies
\begin{equation}
\label{eq:dyn_input_penalty_lower}
        J_{\beta,r}\geq \frac13+\frac{r}{3}+\frac{\beta}{12}
\end{equation}
whenever \(0\leq r\leq \beta/2\). Consequently, the static ReLU controller strictly beats every admissible dynamic linear controller whenever
\begin{equation}
\label{eq:r_threshold}
        0\leq r<\frac{8\beta}{32+9\beta}.
\end{equation}
For \(\beta=1\), this condition is \(r<8/41\).
\end{theorem}

\begin{proof}
For a dynamic linear controller, using only the costs at \(k=0\) and the state cost at \(k=1\),
\begin{align}
J_{\beta,r}
&\geq \frac13 + \E[r(d+\kappa x_0)^2]
 +\beta\E\left[\left(\left(\frac32+\kappa\right)x_0+d\right)^2\right] \\
&=\frac13 +(r+\beta)d^2+\frac{1}{3}\left(r\kappa^2+\beta\left(\frac32+\kappa\right)^2\right).
\end{align}
Minimization over \(|d|+|\kappa|\leq1\) gives \(d=0\). If \(r\leq\beta/2\), the constrained minimizer in \(\kappa\) is \(\kappa=-1\), yielding \eqref{eq:dyn_input_penalty_lower}.

For the ReLU controller, \(x_2=0\). We already have \(\E[x_1^2]=1/36\). Also,
\begin{equation}
        \E[u_0^2]=\int_0^{2/3}(1.5x)^2\,dx+\int_{2/3}^{1}1\,dx=\frac59,
\end{equation}
and, since \(u_1=-1.5x_1\),
\begin{equation}
        \E[u_1^2]=\frac94\E[x_1^2]=\frac{1}{16}.
\end{equation}
This proves \eqref{eq:nn_input_penalty_cost}. The inequality \eqref{eq:nn_input_penalty_cost}\(<\)\eqref{eq:dyn_input_penalty_lower} is equivalent to \eqref{eq:r_threshold}; this threshold is always below \(\beta/2\), so the lower bound applies.
\end{proof}

\subsection{Comparison with Static Output Feedback}
\label{sec:sof}
Static output feedback (SOF) can describe different controller classes. We distinguish three cases.

\begin{proposition}[No dominance over arbitrary nonlinear SOF]
\label{prop:no_arbitrary_sof}
Let \(\mcC_{\rm nlSOF}\) be the class of all admissible static nonlinear output-feedback laws \(u=\mu(y)\), and let \(\mcC_{\rm INC}\) be the subclass representable by \eqref{eq:controller}. Then
\begin{equation}
        \mcC_{\rm INC}\subseteq\mcC_{\rm nlSOF}.
\end{equation}
Consequently, no controller in \(\mcC_{\rm INC}\) can be proved strictly superior to the best controller in \(\mcC_{\rm nlSOF}\) without imposing additional restrictions on the nonlinear SOF class.
\end{proposition}

\begin{proof}
For every well-posed set of matrices in \eqref{eq:controller}, the input \(u\) is a single-valued function of the current measurement \(y\). Hence it is a static nonlinear output-feedback law. The inclusion implies \(\inf_{\mu\in\mcC_{\rm nlSOF}}J(\mu)\leq\inf_{\pi\in\mcC_{\rm INC}}J(\pi)\) for any common cost and constraints.
\end{proof}

Thus meaningful strict comparison should be made against a proper subclass. The simplest is linear static output feedback.

\begin{theorem}[Separation from linear static output feedback]
\label{thm:lsof_sep}
Consider the scalar plant \eqref{eq:scalar_plant} with measured output \(y=x\). Let the benchmark class be all linear static output-feedback laws \(u=Fy\) satisfying \(|Fx|\leq1\) for all \(|x|\leq1\). Imposing the same trajectory-wise input admissibility can only restrict this benchmark class, and the minimizing gain identified below is trajectory-wise admissible. Then, for every \(\beta\in(0,1]\),
\begin{equation}
\label{eq:lsof_optimal_cost}
        \inf_F \E\sum_{k=0}^{\infty}\beta^k x_k^2
        =\frac{1}{3(1-\beta/4)}.
\end{equation}
The static two-ReLU output-feedback law \(u=-\sat(1.5y)\) has cost \(1/3+\beta/36\), and hence has strictly smaller cost than every admissible linear static output-feedback controller.
\end{theorem}

\begin{proof}
The input constraint gives \(|F|\leq1\). The closed-loop state is \(x_k=(1.5+F)^kx_0\). Therefore, for gains whose discounted series is finite,
\begin{equation}
        J_\beta(F)=\frac{\E[x_0^2]}{1-\beta(1.5+F)^2}.
\end{equation}
Gains for which the series diverges have infinite cost and cannot improve the infimum. Among the finite-cost gains, the denominator is maximized, and the cost minimized, by minimizing \(|1.5+F|\) over \([-1,1]\), which gives \(F=-1\). This gain keeps \(|x_k|\leq1\) and \(|u_k|\leq1\) for all \(x_0\in[-1,1]\). This proves \eqref{eq:lsof_optimal_cost}. The ReLU cost is given by Theorem~\ref{thm:longrun_dyn}. The gap is
\begin{equation}
        \frac{1}{3(1-\beta/4)}-\left(\frac13+\frac{\beta}{36}\right)
        =\frac{\beta(8+\beta)}{144(1-\beta/4)}>0.
\end{equation}
\end{proof}

A truly output-feedback version is obtained by embedding the scalar unstable coordinate in a larger system.

\begin{corollary}[An output-feedback embedding]
\label{cor:genuine_of}
Consider
\begin{equation}
\begin{bmatrix}x_{1,k+1}\\x_{2,k+1}\end{bmatrix}
=
\begin{bmatrix}1.5&0\\0&\alpha\end{bmatrix}
\begin{bmatrix}x_{1,k}\\x_{2,k}\end{bmatrix}
+
\begin{bmatrix}1\\0\end{bmatrix}u_k,
\qquad y_k=x_{1,k},
\end{equation}
with \(|\alpha|<1\). If the performance index contains \(\sum_k\beta^k x_{1,k}^2\), and any term involving \(x_2\) is independent of \(u\), then the static implicit output-feedback controller \(u=-\sat(1.5y)\) 
has strictly smaller cost than every admissible linear static output-feedback law \(u=Fy\), in the same sense as Theorem~\ref{thm:lsof_sep}.
\end{corollary}

\subsection{A General Long-Run Comparison Certificate}
\label{sec:general_certificate}
The scalar separation rests on a general principle: a finite-dimensional dynamic linear output-feedback controller has an affine first move in the measured output, whereas static ReLU controllers can implement saturated continuous piecewise-affine policies. This yields a lower bound for the linear class that can be compared with an LMI-certified upper bound for an admissible implicit neural controller.

Let \(\mcX\subset\R^n\) be a compact initial set, \(x_0\sim\mu\) supported on \(\mcX\), \(y=C_px\), and \(\mcU=\{u:H_uu\leq h_u\}\) be a compact input polytope. Consider the state-only discounted cost
\begin{equation}
\label{eq:general_cost}
        J_\beta(\pi)=\E\sum_{k=0}^{\infty}\beta^k x_k^\top Qx_k,
        \qquad Q\succeq0.
\end{equation}
For a finite-dimensional LTI dynamic output-feedback controller with fixed initial controller state,
\begin{equation}
        \xi_{k+1}=A_c\xi_k+B_cy_k,
        \qquad
        u_k=C_c\xi_k+D_cy_k,
\end{equation}
the first action has the form
\begin{equation}
\label{eq:affine_first_output}
        u_0=D_yC_px_0+d,
\end{equation}
where \(D_y=D_c\) and \(d=C_c\xi_0\). If the benchmark is full-state feedback, one sets \(C_p=I\). Alternatively, replacing \(D_yC_p\) by an arbitrary free matrix \(D\) only relaxes the lower-bound problem and is conservative in favor of the linear benchmark.
Admissibility at the first step requires \((D_y,d)\in\mcA_{\rm aff}^y\), where
\begin{equation}
\label{eq:aff_admissible}
\begin{aligned}
        \mcA_{\rm aff}^y\triangleq
        \{(D_y,d):\;&H_u(D_yC_px+d)\leq h_u,\\
        &\forall x\in\mcX\}.
\end{aligned}
\end{equation}
Define the affine first-move lower bound
\begin{equation}
\label{eq:Jlin_star}
\begin{aligned}
        L_{\rm lin}\triangleq&\;\E[x_0^\top Qx_0] \\
        &+\beta\min_{(D_y,d)\in\mcA_{\rm aff}^y}
        \E\left[\norm{A_px_0+B_p(D_yC_px_0+d)}_Q^2\right].
\end{aligned}
\end{equation}
Because all future state costs are nonnegative, every input-admissible dynamic linear output-feedback controller has \(J_\beta\geq L_{\rm lin}\).

\begin{theorem}[Certified long-run superiority]
\label{thm:general_superiority}
Let \(\pi_{\rm NN}\) be a well-posed static implicit neural controller. Assume that it is admissible on the closed-loop trajectories issuing from \(\mcX\), namely
\begin{equation}
\label{eq:nn_admissibility_general}
        H_u\pi_{\rm NN}(C_px_k)\leq h_u,
        \qquad k\geq0,
        \quad x_0\in\mcX,
\end{equation}
either by construction or by an independent regional certificate. Suppose Theorem~\ref{thm:performance}, with \(R=0\), returns a matrix \(P\succ0\) such that
\begin{equation}
\label{eq:superiority_cert}
        \Tr(PS_0)<L_{\rm lin},
\end{equation}
where \(S_0=\E[x_0x_0^\top]\). Then \(\pi_{\rm NN}\) has strictly smaller expected discounted cost than every admissible finite-dimensional dynamic linear output-feedback controller in the benchmark class.
\end{theorem}

\begin{proof}
By admissibility, \(\pi_{\rm NN}\) belongs to the same constrained comparison problem. By Theorem~\ref{thm:performance}, \(\E[J_\beta(\pi_{\rm NN})]\leq\Tr(PS_0)\). By construction of \(L_{\rm lin}\), every admissible dynamic linear output-feedback controller has \(\E[J_\beta]\geq L_{\rm lin}\). The strict inequality \eqref{eq:superiority_cert} proves the claim.
\end{proof}

\begin{remark}[Computation]
If \(\mcX\) is a polytope with vertices \(v_i\), then \eqref{eq:aff_admissible} is equivalent to finitely many inequalities \(H_u(D_yC_pv_i+d)\leq h_u\). The minimization in \eqref{eq:Jlin_star} is then a convex quadratic program for empirical distributions, for distributions represented by moments, and for uniform distributions approximated by quadrature or samples. If the relaxed full-state form \(u_0=Dx_0+d\) is used, the resulting lower bound is no larger than the output-feedback one and therefore still gives a valid, benchmark-favorable comparison.
\end{remark}

\begin{proposition}[Static ReLU representation of PWA policies]
\label{prop:pwa_relu}
Every continuous piecewise-affine policy with finitely many regions can be represented exactly by a finite ReLU network. Consequently, it can be represented as a well-posed static implicit neural controller.
\end{proposition}

\begin{proof}
Each scalar continuous piecewise-affine function can be expressed as the difference of two convex piecewise-affine functions. Each convex piecewise-affine function is the maximum of finitely many affine functions. The identity
\begin{equation}
        \max\{r,s\}=r+\sigma(s-r)
\end{equation}
shows that finite maxima of affine functions can be represented by ReLU networks. Vector outputs are handled componentwise. A feedforward ReLU network is embedded in \eqref{eq:controller} by stacking its layer activations into \(\eta\), yielding a strictly block-triangular matrix \(W\) and hence \(\lambda_{\pf}(|W|)=0\).
\end{proof}

\subsection{General Linear Static Output-Feedback Comparison}
\label{sec:general_sof}
For a general LTI plant with output \(y=C_px\), the implicit output-feedback controller is
\begin{equation}
        \eta=\phi(W\eta+Uy),\qquad u=Ky+V\eta.
\end{equation}
The LMI definitions in Sections~\ref{sec:stability}--\ref{sec:performance} already include this case through \(Z=[UC_p\;W]\), \(F=[A_p+B_pKC_p\;B_pV]\), and \(G=[KC_p\;V]\).

For a linear static output-feedback controller \(u=Fy=FC_px\), the discounted quadratic cost, when the closed-loop matrix \(A_F=A_p+B_pFC_p\) is Schur in the discounted sense, is
\begin{equation}
        J_{\beta,F}(x_0)=x_0^\top P_Fx_0,
\end{equation}
where \(P_F\) solves
\begin{equation}
\label{eq:PF_lyap}
        P_F=Q+C_p^\top F^\top RF C_p+\beta A_F^\top P_FA_F.
\end{equation}
Let \(\mcF_{\rm SOF}\) be the set of linear static output-feedback gains satisfying the prescribed stability and input constraints. If the implicit neural controller admits the certified upper bound \(\Tr(P_{\rm NN}S_0)\) from Theorem~\ref{thm:performance}, and if one can compute or lower-bound
\begin{equation}
        L_{\rm SOF}\leq \inf_{F\in\mcF_{\rm SOF}}\Tr(P_FS_0),
\end{equation}
then
\begin{equation}
        \Tr(P_{\rm NN}S_0)<L_{\rm SOF}
\end{equation}
certifies strict superiority over all gains in \(\mcF_{\rm SOF}\). The scalar theorem of Section~\ref{sec:sof} is the closed-form instance of this comparison.

\section{Numerical Examples}
\label{sec:numerics}
This section reports three reproducible numerical illustrations. The first is a fully certified example with a nonzero hidden interconnection \(W\). The second visualizes the analytic scalar separation of Theorem~\ref{thm:longrun_dyn}. The third is a two-state constrained example; it is useful diagnostically, but the performance number for the saturated law is a sampled estimate and is labeled as such.

\subsection{A certified implicit controller}
\label{sec:numerics_implicit}
Consider the scalar plant
\begin{equation}
\label{eq:implicit_example_plant}
        x_{k+1}=1.2x_k+u_k,
        \qquad y_k=x_k,
\end{equation}
with discounted stage cost \(x_k^2+0.05u_k^2\) and \(\beta=0.98\). The controller has two hidden units, leaky-ReLU activation \(\phi(r)=\max\{r,0.5r\}\), and
\begin{equation}
\label{eq:implicit_example_controller}
\begin{aligned}
        W&=\begin{bmatrix}0.3&0.2\\0.2&0.3\end{bmatrix},
        & U&=\begin{bmatrix}1\\-1\end{bmatrix},\\
        V&=\begin{bmatrix}-0.7&0.7\end{bmatrix},
        & K&=0.
\end{aligned}
\end{equation}
Thus \(\lambda_{\pf}(|W|)=\norm{W}_\infty=0.5<1\), so the hidden fixed point is unique. Over \(x\in[-1,1]\), fixed-point iteration from zero reached tolerance \(10^{-10}\) in at most 25 iterations.

The leaky ReLU satisfies the sector condition with \(D_0=0.5I\), \(D_1=I\). The feasibility problem is two-dimensional in this example when \(\Lambda=\lambda I_2\) is imposed. A direct scalar search over \((P,\lambda)\), followed by a double-precision eigenvalue check of the rounded matrices, gives
\begin{equation}
\label{eq:implicit_example_certificate}
\begin{gathered}
        P=1.3006,\qquad \Lambda=0.7749 I_2,\\
        \lambda_{\max}(\mathrm{LMI})=-2.23\times10^{-5}.
\end{gathered}
\end{equation}
Hence
\begin{equation}
        J_\beta(x_0)\leq 1.3006 x_0^2,
        \qquad
        \E[J_\beta]\leq0.4335
\end{equation}
for \(x_0\sim{\rm Unif}[-1,1]\). Table~\ref{tab:implicit_reproducibility} gives the reproducibility data used for this certificate. No training optimizer is used in this example; it is a feasibility demonstration for the analysis certificate with \(W\neq0\).

\begin{table}[t]
\caption{Reproducibility data for the nonzero-\(W\) certificate.}
\label{tab:implicit_reproducibility}
\centering
\small
\begin{tabular}{@{}p{0.43\columnwidth}p{0.49\columnwidth}@{}}
\toprule
quantity & value \\
\midrule
\(\lambda_{\pf}(|W|)\) & \(0.5\) \\
fixed-point method & Picard iteration from zero \\
fixed-point tolerance & \(10^{-10}\) \\
max. iterations on \([-1,1]\) & \(25\) \\
LMI search & direct 2-D search; \(\Lambda=\lambda I_2\) \\
reported rounded values & \(P=1.3006,\;\lambda=0.7749\) \\
max. eigenvalue after rounding & \(-2.23\times10^{-5}\) \\
\bottomrule
\end{tabular}
\end{table}

The exact cost can also be computed analytically. For \(x>0\), the active slope matrix is \(D_+=\diag(1,0.5)\), and
\begin{equation}
        \eta=(I-D_+W)^{-1}D_+Ux
        =\begin{bmatrix}1.3043\\-0.4348\end{bmatrix}x.
\end{equation}
Thus \(u=V\eta=-1.2174x\) and \(x^+=-0.01739x\). For \(x<0\), the active slope matrix is \(D_-=\diag(0.5,1)\), which yields the same effective gain by symmetry. Therefore
\begin{equation}
        J_\beta(x_0)=
        \frac{1+0.05(1.2174)^2}{1-0.98(0.01739)^2}\,x_0^2
        =1.0744x_0^2.
\end{equation}
With the same \(U,V,K\) but \(W=0\), the effective gain is \(-1.05\), the closed-loop multiplier is \(0.15\), and the exact cost is \(1.0789x_0^2\); the corresponding LMI value is \(P=1.3570\), hence \(\E[J_\beta]\leq0.4523\). The improvement is modest, however, the point of the example is not to establish a broad performance advantage of implicit hidden feedback, but to demonstrate that the LMI/IQC machinery certifies a truly implicit fixed-point controller, not only feedforward special cases (those with $W=0$).

\subsection{Scalar separation visualized}
\begin{figure*}[htb]
\centering
\includegraphics[width=0.94\textwidth]{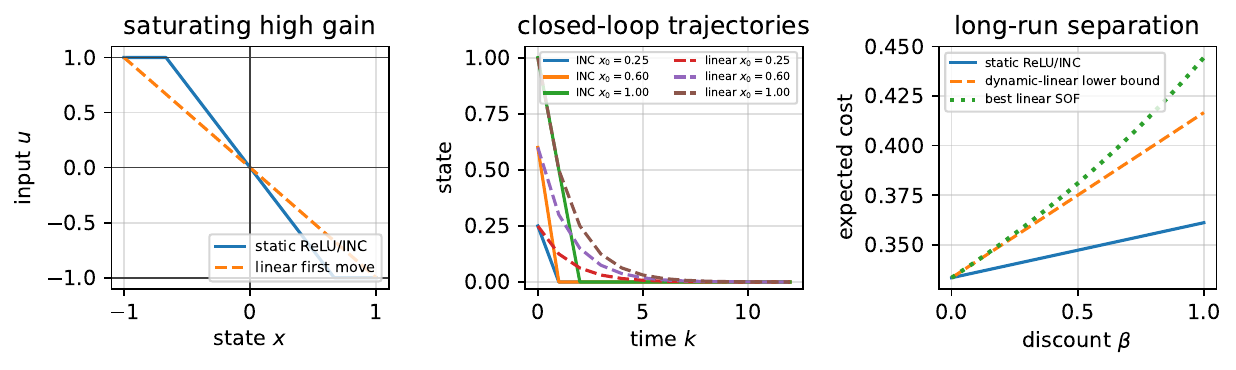}
\caption{Analytic scalar separation. Left: the static ReLU/INC law uses high gain near the origin and saturates at the admissible input limits. Middle: representative closed-loop trajectories. Right: exact discounted costs and lower bounds as functions of the discount factor.}
\label{fig:scalar_separation}
\end{figure*}

Figure~\ref{fig:scalar_separation} illustrates the mechanism behind Theorem~\ref{thm:longrun_dyn}. The nonlinear controller uses the linear gain \(-1.5\) near the origin and saturates at the input limits near the edge of the interval. The best admissible linear first move has slope \(-1\). The middle panel shows that the static ReLU controller reaches the origin in at most two steps for the displayed initial conditions, while the linear feedback contracts more slowly. The right panel compares the exact discounted costs: the static ReLU/INC cost remains strictly below both the dynamic-linear lower bound and the best linear static output-feedback cost for every discount factor shown.

\subsection{Two-state constrained LTI example}
We next consider
\begin{equation}
\label{eq:two_state_plant}
        A_p=\begin{bmatrix}1.10&1.00\\0&1.02\end{bmatrix},
        \qquad
        B_p=\begin{bmatrix}0.20\\0.50\end{bmatrix},
\end{equation}
with \(|u_k|\leq1\), \(x_0\sim{\rm Unif}([-0.8,0.8]^2)\), discount \(\beta=0.98\), and stage cost \(x_k^\top x_k+0.05u_k^2\). The nonlinear controller is the saturated discounted-LQR law
\begin{equation}
\label{eq:two_state_saturated_lqr}
        u_k=\sat(K_{\rm LQR}x_k),
        \qquad K_{\rm LQR}=\begin{bmatrix}-1.3749&-2.6490\end{bmatrix}.
\end{equation}
This law is a static ReLU controller, hence a special case of the implicit architecture with \(W=0\), because
\begin{equation}
        \sat(v)=\sigma(v+1)-\sigma(v-1)-1.
\end{equation}

As a linear benchmark, we performed a dense two-variable search over static gains \(F\) satisfying
\begin{equation}
        |Fx|\leq1,
        \qquad \forall x\in[-0.8,0.8]^2,
\end{equation}
or equivalently \(0.8(|F_1|+|F_2|)\leq1\). For each feasible linear gain, the discounted cost was computed exactly by the Lyapunov equation
\begin{equation}
        P_F=I+0.05F^\top F+\beta(A_p+B_pF)^\top P_F(A_p+B_pF),
\end{equation}
and the expected cost was \(\Tr(P_FS_0)\), where \(S_0=(0.8)^2I/3\). The best gain found in this dense search over the feasible diamond is
\begin{equation}
        F^\star=\begin{bmatrix}-0.3535&-0.8965\end{bmatrix},
        \qquad J(F^\star)=3.3472.
\end{equation}
Scaling the unconstrained LQR gain just enough to satisfy the same initial-set input bound gives cost \(3.5639\). For the reported linear gain, the numerical check \(\max_{0\leq k\leq100}0.8\norm{F(A_p+B_pF)^k}_1=1.0000\) is attained at \(k=0\); this is a consistency check for the plotted initial-set trajectories, not a proof of global optimality of the search. By Monte Carlo integration with \(2\times10^5\) initial states, the saturated static INC law gives
\begin{equation}
        J_{\rm INC}=1.3579\pm0.0071,
\end{equation}
where the interval is two standard errors. The improvement is not due to a larger sampled input: both controllers respect \(|u_k|\leq1\) on the simulated grid shown in Fig.~\ref{fig:two_state_numerics}. The gain comes from using the unconstrained LQR action when it is feasible and clipping only where the actuator bound is active. Since the nonlinear cost is estimated by sampling and the linear gain is the best found in a dense search rather than a certified global optimizer, this experiment should be read as an illustration of the constrained gain-scheduling mechanism, rather than as a theorem comparable to Theorem~\ref{thm:longrun_dyn}.

\begin{figure*}[t]
\centering
\includegraphics[width=0.88\textwidth]{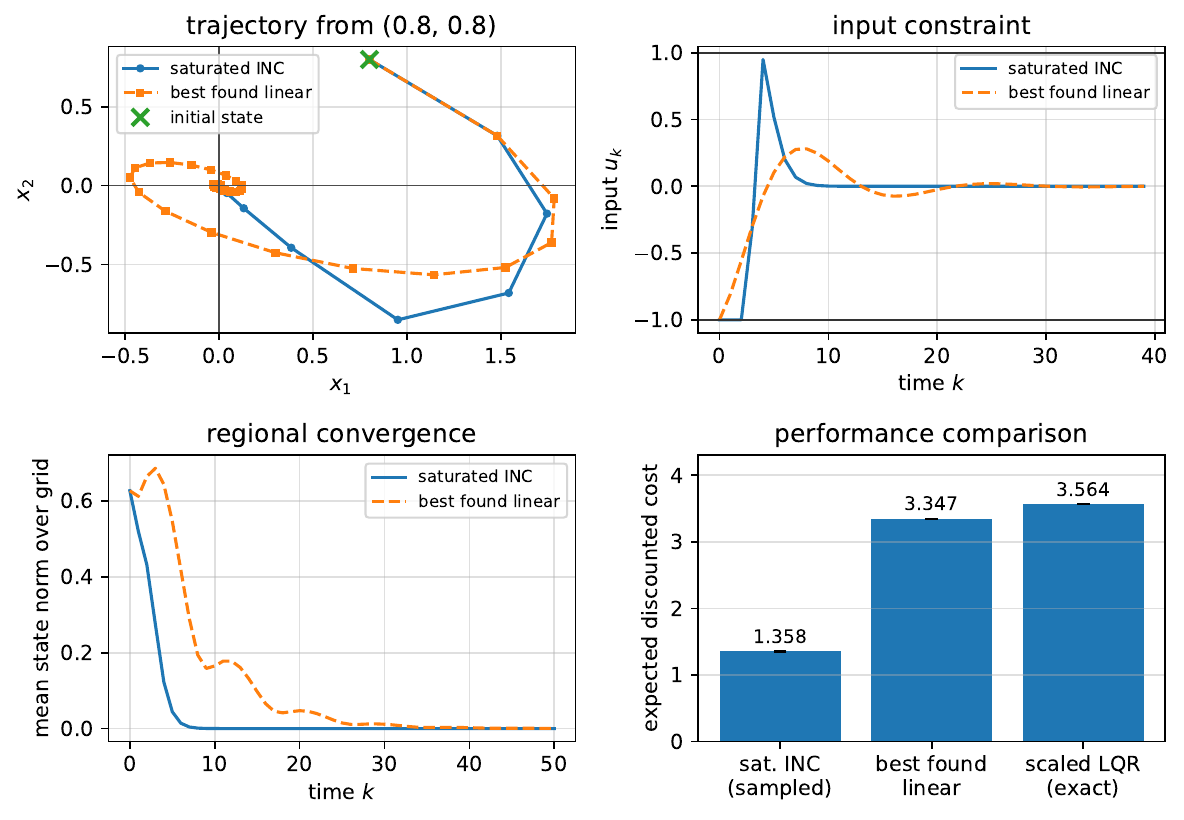}
\caption{Two-state constrained LTI example. Top left: representative state trajectory from \((0.8,0.8)\). Top right: input traces, showing saturation of the INC and feasibility of the linear benchmark. Bottom left: mean state norm over a uniform grid of initial states. Bottom right: expected discounted cost. The saturated static INC has lower sampled cost than the best static linear gain found in the dense search and a scaled LQR gain on this benchmark.}
\label{fig:two_state_numerics}
\end{figure*}


\section{Conclusion}
This paper presented implicit neural controllers as state-space-like objects for certified neural feedback. The central point is that the implicit representation turns a neural controller into a familiar control interconnection: trainable linear equations closed through a known static nonlinearity. This exposes the hidden algebraic loop, gives checkable well-posedness conditions, and allows Lyapunov/IQC certificates to be written directly on the plant-controller interconnection.

For the static LTI setting studied here, Perron--Frobenius and norm conditions ensure that the controller fixed point is well posed. Static and incremental sector IQCs lead to LMI tests for exponential stability and discounted infinite-horizon quadratic performance, including biased controllers around arbitrary equilibria. The constrained separation theorems show that static nonlinear neural feedback can be provably superior to dynamic linear feedback in a long-run performance index. The key structural advantage is actuator-aware gain scheduling: a ReLU controller can use high gain near the origin and saturate near the boundary, whereas a globally admissible linear controller must use a single conservative affine move.

The presented theory clearly  has its limitations. The LMI/IQC conditions are analysis certificates, not a convex synthesis theory. The train--certify--retrain loop can propose useful controllers, but feasibility is checked only after the controller is fixed; hence, it is an heuristic design approach, with no guarantee of correct outcome. Also, the global sector LMIs can be conservative, especially for ReLU or saturation when the open-loop plant is unstable: regional multipliers, activation-pattern arguments, or explicit input-admissibility checks are then needed. 
Several directions also remain open: scalable design conditions with reduced conservatism, regional certificates for implicit controllers with \(W\neq0\), disturbance and robustness margins, observer-plus-implicit-readout architectures, and extensions to dynamic implicit neural controllers.

\appendices

\section{Bias-Augmented ReLU Realization of Saturation}
\label{app:saturation_relu}
The saturation used in \eqref{eq:relu_sat_controller} is realized as an implicit model with two hidden units:
\begin{equation}
        \eta_1=\sigma(1.5x+1),\qquad
        \eta_2=\sigma(1.5x-1),
\end{equation}
\begin{equation}
        u=1-\eta_1+\eta_2.
\end{equation}
Thus \(W=0\), \(U=[1.5\;1.5]^\top\), \(b=[1\;-1]^\top\), \(V=[-1\;1]\), \(K=0\), and \(d=1\). Since \(W=0\), well-posedness is immediate.

\section{Equivalent IQC Matrix Form}
The generalized sector IQC can be written as
\begin{equation}
        2(\eta-D_0z)^\top\Lambda(D_1z-\eta)\geq0.
\end{equation}
Equivalently,
\begin{equation}
        \begin{bmatrix}z\\\eta\end{bmatrix}^\top
        \begin{bmatrix}
        -2D_0\Lambda D_1 & D_0\Lambda+\Lambda D_1\\
        \Lambda D_1+D_0\Lambda & -2\Lambda
        \end{bmatrix}
        \begin{bmatrix}z\\\eta\end{bmatrix}\geq0,
\end{equation}
where the diagonal matrices commute. Since \(z=Zs\) and \(\eta=Hs\), this is exactly \(s^\top M_\Lambda s\geq0\). In the special case \(D_0=0\), \(D_1=I\), this reduces to
\begin{equation}
        \begin{bmatrix}z\\\eta\end{bmatrix}^\top
        \begin{bmatrix}0&\Lambda\\\Lambda&-2\Lambda\end{bmatrix}
        \begin{bmatrix}z\\\eta\end{bmatrix}\geq0.
\end{equation}


\begin{thebibliography}{99}

\bibitem{elghaoui2020implicit}
L.~El~Ghaoui, F.~Gu, B.~Travacca, A.~Askari, and A.~Y.~Tsai,
``Implicit deep learning,''
\emph{SIAM Journal on Mathematics of Data Science}, vol. 3, no. 3, pp. 930--958, 2021.

\bibitem{narendra1990identification}
K.~S. Narendra and K.~Parthasarathy, ``Identification and control of dynamical systems using neural networks,'' \emph{IEEE Transactions on Neural Networks}, vol.~1, no.~1, pp. 4--27, 1990.

\bibitem{miller1995neural}
W.~T. Miller, P.~J. Werbos, and R.~S. Sutton, Eds., \emph{Neural Networks for Control}. Cambridge, MA, USA: MIT Press, 1995.

\bibitem{lewis1999neural}
F.~L. Lewis, S.~Jagannathan, and A.~Yesildirak, \emph{Neural Network Control of Robot Manipulators and Nonlinear Systems}. London, U.K.: Taylor \& Francis, 1999.

\bibitem{boettcher2026control}
L.~B{\"o}ttcher, ``Control of dynamical systems with neural networks,'' \emph{Nonlinear Dynamics}, vol.~114, Art. no.~79, 2026, doi: 10.1007/s11071-025-11937-z.

\bibitem{mnih2015human}
V.~Mnih, K.~Kavukcuoglu, D.~Silver, A.~A. Rusu, J.~Veness, M.~G. Bellemare, A.~Graves, M.~Riedmiller, A.~K. Fidjeland, G.~Ostrovski, \emph{et al.}, ``Human-level control through deep reinforcement learning,'' \emph{Nature}, vol. 518, pp. 529--533, 2015.

\bibitem{degrave2022magnetic}
J.~Degrave, F.~Felici, J.~Buchli, M.~Neunert, B.~Tracey, F.~Carpanese, T.~Ewalds, R.~Hafner, A.~Abdolmaleki, D.~de~las~Casas, \emph{et al.}, ``Magnetic control of tokamak plasmas through deep reinforcement learning,'' \emph{Nature}, vol. 602, pp. 414--419, 2022.

\bibitem{saintdonat1991neural}
J.~Saint-Donat, N.~Bhat, and T.~J. McAvoy, ``Neural net based model predictive control,'' \emph{International Journal of Control}, vol.~54, no.~6, pp. 1453--1468, 1991.

\bibitem{draeger1995mpc}
A.~Draeger, S.~Engell, and H.~Ranke, ``Model predictive control using neural networks,'' \emph{IEEE Control Systems Magazine}, vol.~15, no.~5, pp. 61--66, 1995.

\bibitem{amos2017optnet}
B. Amos and J. Z. Kolter, ``OptNet: Differentiable optimization as a layer in neural networks,'' in \emph{Proc. International Conference on Machine Learning}, 2017, pp. 136--145.

\bibitem{amos2018differentiablempc}
B.~Amos, I.~Jimenez, J.~Sacks, B.~Boots, and J.~Z. Kolter, ``Differentiable MPC for end-to-end planning and control,'' in \emph{Advances in Neural Information Processing Systems}, 2018, pp. 8289--8300.

\bibitem{chen2018neuralode}
R. T. Q. Chen, Y. Rubanova, J. Bettencourt, and D. Duvenaud, ``Neural ordinary differential equations,'' in \emph{Advances in Neural Information Processing Systems}, 2018.

\bibitem{kaiser2018sindy}
E.~Kaiser, J.~N. Kutz, and S.~L. Brunton, ``Sparse identification of nonlinear dynamics for model predictive control in the low-data limit,'' \emph{Proceedings of the Royal Society A}, vol. 474, no. 2219, 2018, Art. no. 20180335.

\bibitem{bai2019deep}
S.~Bai, J.~Z. Kolter, and V.~Koltun, ``Deep equilibrium models,'' in \emph{Advances in Neural Information Processing Systems}, 2019.

\bibitem{boyd1994lmi}
S. Boyd, L. El Ghaoui, E. Feron, and V. Balakrishnan, \emph{Linear Matrix Inequalities in System and Control Theory}. Philadelphia, PA, USA: SIAM, 1994.

\bibitem{megretski1997iqc}
A. Megretski and A. Rantzer, ``System analysis via integral quadratic constraints,'' \emph{IEEE Transactions on Automatic Control}, vol. 42, no. 6, pp. 819--830, 1997.

\bibitem{khalil2002nonlinear}
H. K. Khalil, \emph{Nonlinear Systems}, 3rd ed. Upper Saddle River, NJ, USA: Prentice Hall, 2002.

\bibitem{dullerud2000robust}
G. E. Dullerud and F. Paganini, \emph{A Course in Robust Control Theory: A Convex Approach}. New York, NY, USA: Springer, 2000.

\bibitem{raghunathan2018sdp}
A. Raghunathan, J. Steinhardt, and P. Liang, ``Semidefinite relaxations for certifying robustness to adversarial examples,'' in \emph{Advances in Neural Information Processing Systems}, 2018, pp. 10877--10887.

\bibitem{fazlyab2022safety}
M.~Fazlyab, M.~Morari, and G.~J. Pappas, ``Safety verification and robustness analysis of neural networks via quadratic constraints and semidefinite programming,'' \emph{IEEE Transactions on Automatic Control}, vol.~67, no.~1, pp. 1--15, 2022.

\bibitem{yin2022qc}
H.~Yin, P.~Seiler, and M.~Arcak, ``Stability analysis using quadratic constraints for systems with neural network controllers,'' \emph{IEEE Transactions on Automatic Control}, vol.~67, no.~4, pp. 1980--1987, 2022.

\bibitem{junnarkar2024dissipativity}
N.~Junnarkar, M.~Arcak, and P.~Seiler, ``Synthesizing neural network controllers with closed-loop dissipativity guarantees,'' \emph{Automatica}, to appear, arXiv:2404.07373, 2026.

\bibitem{gilbert1991linear}
E.~G. Gilbert and K.~T. Tan, ``Linear systems with state and control constraints: The theory and application of maximal output admissible sets,'' \emph{IEEE Transactions on Automatic Control}, vol.~36, no.~9, pp. 1008--1020, 1991.

\bibitem{tarbouriech2011stability}
S.~Tarbouriech, G.~Garcia, J.~M.~G. da~Silva~Jr., and I.~Queinnec, \emph{Stability and Stabilization of Linear Systems with Saturating Actuators}. London, U.K.: Springer, 2011.

\bibitem{bemporad2002explicit}
A. Bemporad, M. Morari, V. Dua, and E. N. Pistikopoulos, ``The explicit linear quadratic regulator for constrained systems,'' \emph{Automatica}, vol. 38, no. 1, pp. 3--20, 2002.

\bibitem{rawlings2017mpc}
J.~B. Rawlings, D.~Q. Mayne, and M.~M. Diehl, \emph{Model Predictive Control: Theory, Computation, and Design}, 2nd ed. Madison, WI, USA: Nob Hill, 2017.

\bibitem{boyd2004convex}
S. Boyd and L. Vandenberghe, \emph{Convex Optimization}. Cambridge, U.K.: Cambridge University Press, 2004.

\bibitem{horn2012matrix}
R. A. Horn and C. R. Johnson, \emph{Matrix Analysis}, 2nd ed. Cambridge, U.K.: Cambridge University Press, 2012.

\end{thebibliography}
\end{document}